\documentclass[aps,twocolumn,superscriptaddress,nofootinbib,a4paper,longbibliography, pra]{revtex4-2}

\usepackage{times}
\usepackage{epsfig}
\usepackage{amsfonts}
\usepackage{amsmath}
\usepackage{amsthm}
\usepackage{amssymb}					
\usepackage{amsthm}
\usepackage{dsfont}
\usepackage{bm}
\usepackage{mathtools}
\usepackage{color}
\usepackage{multirow}
\usepackage[normalem]{ulem}
\newcommand{\stkout}[1]{\ifmmode\text{\sout{\ensuremath{#1}}}\else\sout{#1}\fi}
\usepackage{latexsym}
\usepackage{mathrsfs}
\usepackage{natbib}
\usepackage{verbatim}
\usepackage[T1]{fontenc}

\usepackage{graphicx}
\usepackage{xcolor}
\usepackage{physics}
\usepackage{soul}
\usepackage{bm,color,bbm}
\usepackage{hyperref, mathtools,graphicx,soul}
\usepackage{subfigure} 
\usepackage{amsmath,amssymb,amsthm}
\usepackage{amsmath,blkarray}
\usepackage{braket}
\usepackage{color}
\usepackage{mathtools}
\usepackage{natbib}
\usepackage{slashbox}
\usepackage{subcaption}
\allowdisplaybreaks

\newcommand{\beq}{\begin{equation}}
\newcommand{\eeq}{\end{equation}}
\newcommand{\bqa}{\begin{eqnarray}}
\newcommand{\eqa}{\end{eqnarray}}

\definecolor{maroon}{rgb}{0.7,0,0}

\definecolor{ngreen}{rgb}{0.3,0.7,0.3}

\definecolor{golden}{rgb}{0.8,0.6,0.1}

\definecolor{npurple}{rgb}{0.3,0,0.6}

\newcommand{\be}{\begin{equation}}
\newcommand{\ee}{\end{equation}}
\newcommand{\ba}{\begin{eqnarray}}
\newcommand{\ea}{\end{eqnarray}}

\newtheorem{definition}{Definition}
\newtheorem{thm}{Theorem}
\newtheorem{cor}{Corollary}
\newtheorem{lem}{Lemma}

\newtheorem{prop}{Proposition}

\begin{document}

\title{Popescu-Rohrlich box fraction of nonobjective information and distinguishing quantum theory}  

\author{Chellasamy Jebarathinam}
   \affiliation{Physics Division, National Center for Theoretical Sciences, National Taiwan University, Taipei 106319, Taiwan}
\begin{abstract}
It is demonstrated that identifying information-theoretic limitations of quantum Bell nonlocality alone cannot completely distinguish quantum theory from generalized nonsignaling theories. To this end, an information-theoretic concept of certifying nonobjective information by the Popescu-Rohrlich box fraction is employed. Furthermore, in the aforementioned demonstration, a partial answer to the question of what distinguishes quantum theory from generalized nonsignaling theories emerges beyond the one provided by the principle of information causality alone.  This is accomplished by demonstrating that postquantum models identified by the information causality are isolated by the emergence of the Popescu-Rohrlich box fraction of nonobjective information in Bell-local boxes of a nonsignaling model, over the other nonsignaling models.  
\end{abstract}

	\maketitle

\section{Introduction} Einstein, Podolsky, and Rosen (EPR) in their famous paper \cite{EPR35} introduced the notion of element of reality to argue that quantum mechanics is incomplete. This argument is based on the fact that the specific entangled state considered does not satisfy the notion of reality by EPR. However, Bell, in his famous paper~\cite{Bel64}, derived an inequality which holds for any physical theory that satisfies the element of reality by EPR and showed that quantum mechanics has nonlocality in the sense that there exists an entangled state that violates his inequality. The framework of the Bell inequality to demonstrate nonlocality of quantum mechanics forms a basis for the identification and development of many central concepts in quantum information theory~\cite{Wer89, Eke91, BBM92, HHH+09,BCP+14,WJD07,  UCN+19, ALM+08,PB11, TBP+24,CAE+14,BCG+22, BCP14,SAP17}.  

Nonlocality of quantum theory in the violation of a Bell inequality does not contradict the nonsignaling principle; however, nonlocality goes beyond quantum mechanics~\cite{PR94}. One of the goals of quantum information theory, motivated by nonlocality beyond quantum mechanics, has been to distinguish quantum theory from a broad class of theories for information processing~\cite{BHK05, CGMP05,  AGM06,BBL+06,LPS+07,NPA07,Bar07,PPK+09,FSA+13,Pop14}. Remarkably, information causality, as a generalization of the nonsignaling principle, was proposed as an information-theoretic principle in identifying the limitation of quantum nonlocality~\cite{PPK+09}. The principle of information causality, however, has only been seen to provide a partial answer to the question of what distinguishes quantum theory~\cite{ABP+09, Pop14}. On the other side, Chiribella, D'Ariano, and Perinotti introduced a different framework of causal probabilistic theories beyond quantum theory~\cite{CDP10} and remarkably identified a postulate~\cite{CDP11} that distinguishes quantum theory from a broad class of theories for information processing.

In contrast to the information-theoretic framework of quantum nonlocality with entanglement~\cite{Wer89, HHH+09,Bus12,BRL+13} or Bell nonlocality~\cite{Bel64, BCP+14, WSS+19}, there are other ways in quantum information theory to characterize nonclassicality in bipartite states for information processing~\cite{MBC+12,ABC16,BDS+18,CG19}.  Quantum discord ~\cite{OZ01,HV01} is one of these ways to characterize the nonclassicality of bipartite states, and nonnull discord has been used to indicate the nonclassicality for quantum information processing using separable states ~\cite{DSC08, DVB10,MBC+12,ABC16}. In~\cite{Per12}, Perinotti defined nonnull discord in the framework of operational probabilistic theories in~\cite{CDP11} using the notion of objective information, which is an extension of the element of reality by EPR. Remarkably, it was shown by Perinotti~\cite{Per12} that nonnull discord is a generic feature of nonclassicality in causal probabilistic theories. On the other hand,  operational nonclassicality of nonnull discord has been characterized by Bell nonclassicality beyond Bell nonlocality, i.e., the Popescu-Rohrlich box fraction of dimensionally restricted nonlocality~\cite{JAS17,Jeb25}, and such operational nonclassicality as a resource for quantum information processing has also been explored~\cite{JKC+24, Jeb25}.   

Despite information-theoretic properties of generalized nonsignaling models were studied substantially in~\cite{BHK05, CGMP05,  AGM06,BBL+06,LPS+07,NPA07,Bar07,PPK+09,FSA+13} which provided insight into the question of what distinguishes quantum theory from generalized nonsignaling theories, the answer to the question still remains open~\cite{Pop14, BCP+14, SGA+18}. To answer this question differently, in this work, I consider nonobjective information in nonnull discord~\cite{Per12} and its certification by the Popescu-Rohrlich box fraction of dimensionally restricted nonlocality~\cite{Jeb25}. Employing this information-theoretic notion of certification of nonobjective information in addition to the information causality principle, I then provide a partial answer to the question of what distinguishes quantum theory from generalized nonsignaling theories. Astonishingly, this answer goes beyond the other partial answers given in~\cite{SBP09,BS09} by the nonlocality swapping and distillation protocols, respectively, and by the information causality alone~\cite{ABP+09}.

\section{Causal probabilistic theories} Here, I review the framework of operational probabilistic theories introduced in~\cite{CDP10} to define a broad class of theories for information processing.  The primitive notions of an operational theory are those of
systems, tests, and events. In this framework, a preparation test is denoted by $\{\rho_i\}_{i \in X} $, which is a collection of preparation events $\rho_i$ of a system $A$, and an observation test is denoted by  $\{a_j\}_{j \in Y} $, which is a collection of observation events $a_j$. Then, in a probabilistic theory, the sequential  composition ($\circ$) of the preparation test with the observation test ($a_j \circ \rho_i $) gives rise to a joint probability distribution: $p(i,j):=(a_j|\rho_i)$, with $p(i,j) \ge 0$ and $\sum_{i \in X}\sum_{j \in Y} p(i,j)=1$. A theory is causal if, for every preparation test $\{\rho_i\}_{i \in X} $ and every observation test  $\{a_j\}_{j \in Y}$ on system $A$, the marginal probability $p_i=\sum_{j \in Y} p(i,j)$ is independent of the choice of the observation test $\{a_j\}_{j \in Y} $. Precisely, if $\{a_j\}_{j \in Y}$ and $\{b_k\}_{k \in Z}$ are two
different observation tests, then one has $p_i=\sum_{k \in Z} (b_k|\rho_i)$. Equivalently, a theory is causal if and only if, for every system $A$, there is a
unique deterministic effect $e_A$. Such an effect can be used to discard the system $A$ in a parallel composition ($\otimes$) of two systems $A$ and $B$, i.e., if $\rho_{AB}$ denotes the state of the composite system $AB$, then $\rho_B= (e_A|\rho_{AB})$.  In this framework, a state $\rho$ of the system $AB$ is separable if it is a convex combination of factorized states, i.e.,
$\rho= \sum_{i \in Z} p_i \rho_i \otimes \sigma_i $, where $\rho_i$ and $\sigma_i$ are states of systems $A$ and $B$, respectively, and  $\{p_i\}_{i \in Z}$ is a probability distribution.

Using the above framework, which has been shown to provide five elementary axioms common to the broad class of theories of information processing~\cite{CDP10}, it was shown in~\cite{CDP11}  that one postulate, namely purification, distinguishes quantum theory from other causal probabilistic theories. From this postulate, it follows that every mixed state $\rho$ of a system $A$ has a purification, $\Psi$, which is an entangled pure state of the composite system $AB$. Furthermore, if there are two purifications $\Psi$ and $\Psi'$ of the mixed state, then these are connected by a reversible transformation on the purified system $B$. This implies that there exists a mixed state of system $A$ that has a nonunique convex decomposition on pure states, which holds in the case of quantum theory since the state space of any system in quantum theory is not a simplex.

Entanglement is not the only nonclassical feature of correlations in causal probabilistic theories. In~\cite{Per12}, it was shown that nonnull discord, which captures nonclassicality of correlations beyond entanglement \cite{OZ01,HV01}, is a generic signature of nonclassicality of causal probabilistic theories. To define null discord in causal probabilistic theories, the notion of objective information, as an extension of the notion of reality by EPR applied to entangled states~\cite{Sch35,WJD07}, was introduced as follows. 
\begin{definition}
   For the given state $\rho$ of a system, a test $\{\mathcal{A}_i\}_{i \in X}$ provides objective information about the state if the test is repeatable and the state is not disturbed by the test, namely, $\mathcal{A} \circ \rho=\rho$ for   $\sum_{i \in X} \mathcal{A}_i$. In other words, $\rho$ encodes the objective information about  the test $\{\mathcal{A}_i\}_{i \in X}$. 
\end{definition}
 Null-discord states are then defined as follows. 
\begin{definition}\label{defOIND}
  In a causal operational probabilistic theory, a bipartite state $\rho$ has null discord if and only if it  is separable and there exists a test $\{\mathcal{A}_k\}_{k \in X}$ on system $A$ that provides complete objective information
about the state $e_B \circ \rho$, and such that $\{\mathcal{A}_k \otimes I_B\}_{k \in X}$ provides objective information on $\rho$.  
\end{definition}
 It then follows that the state $\rho$ has null discord if and only if it
can be expressed as follows:
\begin{equation}\label{nulldiscpt}
    \rho=\sum_{k \in X} q_k \psi_k \otimes \sigma_k,
\end{equation}
where $\{\psi_k\}_{k \in X}$ is a set of jointly perfectly distinguishable
pure states and $\{q_k\}_{k \in X}$ is a probability distribution. This set is perfectly distinguishable in the sense that there exists a discrimination test $\{a_i\}_{i \in X}$ such that 
$a_i \circ \psi_k =\delta_{ik}$. 
Having the notion of null discord defined for causal probabilistic theories, in~\cite{Per12}, a causal probabilistic theory for which the set of all separable states is assumed to have null-discord was introduced. It then follows that in this theory, the set of normalized states for every system is a simplex. Thus, a causal probabilistic theory where all separable states have null discord cannot describe quantum theory. This implies that the occurrence of nonnull discord of separable states in quantum theory is also specific to its state space of nonclassicality~\cite{FAC+10}.   

\section{Generalized nonsignaling theories} In the framework of generalized nonsignaling theories~\cite{BLM+05,MAG06,Bar07,SBP09}, bipartite states are not described as in the framework of causal probabilistic theories~\cite{CDP10} but by bipartite joint probability distributions directly; i.e. probabilities of a pair of results (outputs) given a pair of measurements (inputs). In other words, causal correlations will be
described by ``boxes'' (i.e. input-output devices), which satisfy the nonsignaling principle~\cite{PR94}.
Here, we shall focus on the simplest possible scenario,
namely, the case of two possible measurements for each party
(inputs $x, y \in \{0, 1\}$); each measurement providing a binary
result (outputs $a, b \in \{0, 1\}$). In this case, a box, denoted by $P(ab|A_xB_y)$, is thus described by a set of $16$ joint probabilities.

For any Bell-local box $P_L(ab|A_xB_y)$, there exists a classical state $\lambda$ of the two systems $A$ and $B$ (a.k.a. local-hidden-variable (LHV) or shared
randomnes),  which occurs with probability $p_\lambda \ge 0$ ($\sum_\lambda p_\lambda=1$), such that it can be decomposed as  
\begin{equation} \label{BLb}
    P_L(ab|A_xB_y)= \sum_\lambda p_\lambda P(a|A_x,\lambda) P(b|B_y,\lambda). 
\end{equation}
On the other hand, any Bell nonlocal box cannot be decomposed in this form and violates a Bell inequality. 
The set of Bell-local boxes forms a polytope, which has $16$ vertices (called deterministic boxes),
\begin{equation}
P^{\alpha\beta\gamma\epsilon}_D(ab|A_xB_y)=\left\{
\begin{array}{lr}
1, & a=\alpha x\oplus \beta\\
   & b=\gamma y\oplus \epsilon \\
0 , & \text{otherwise}.\\
\end{array}
\right.   
\label{eq:locdet}
\end{equation}   
Here, $\alpha,\beta,\gamma,\epsilon\in  \{0,1\}$ and  $\oplus$ denotes
addition modulo  $2$.   The local polytope is itself embedded in a
larger polytope, the nonsignaling polytope, which contains all the boxes compatible with the nonsignaling principle. It
has $8$ nonlocal vertices (called Popescu-Rohrlich (PR)
boxes), 
\begin{multline}
  P^{\alpha\beta\gamma}_{PR}(ab|A_xB_y)\\=\left\{
\begin{array}{lr}
\frac{1}{2}, & a\oplus b=x\cdot y \oplus \alpha x\oplus \beta y \oplus \gamma\\ 
0 , & \text{otherwise},\\
\end{array}
\right. \label{NLV}  
\end{multline}
which are all symmetries of the PR
box, $P_{PR}=P^{000}_{PR}$ \cite{PR94}. The set of boxes attainable by quantum mechanics also forms a convex body, although not a polytope. The quantum set is strictly larger than the local polytope - quantum correlations can be Bell nonlocal - but strictly smaller than the nonsignalling polytope.

Any Bell-local box satisfies the Clauser-Horne-Shimony-Holt (CHSH)   inequality
~\cite{CHS+69} and its symmetries, which are given by
\begin{align}
&\mathcal{B}_{\alpha\beta\gamma} := (-1)^\gamma\braket{A_0B_0}+(-1)^{\beta \oplus \gamma}\braket{A_0B_1}\nonumber\\
&+(-1)^{\alpha \oplus \gamma}\braket{A_1B_0}+(-1)^{\alpha \oplus \beta \oplus \gamma \oplus 1} \braket{A_1B_1}\le2, 
\label{BCHSH}
\end{align} 
where $\braket{A_xB_y}=\sum_{ab}(-1)^{a\oplus
  b}P(ab|A_xB_y)$, on the other hand, Bell nonlocal box violates one of these inequalities.   Any given PR box $P^{\alpha\beta\gamma}_{PR}$ violates one of the CHSH inequalities in Eq. (\ref{BCHSH}) to its algebraic maximum, i.e.,  $\mathcal{B}_{\alpha\beta\gamma}=4$ for $P^{\alpha\beta\gamma}_{PR}$.
Quantum correlations violate the CHSH inequality up to $2\sqrt{2}$~\cite{Cir80}, a value known as Tsirelson's
bound.

The PR box implies that nonsignaling as a physical principle does not rule out quantum theory from generalized nonsignaling theories, as pointed out by Popescu and Rohrlich in the seminal paper~\cite{PR94}. With the emergence of Bell nonlocality as a powerful resource for information processing, information-theoretic properties of generalized nonsignaling models were studied to understand Tsirelson's bound of quantum theory~\cite{BBL+06,BS09,PPK+09,BCP+14}. In this direction, 
information causality, as a generalization of the nonsignaling principle, was proposed to single out quantum theory from generalized nonsignaling theories~\cite{PPK+09}.  The information causality is formulated by a generic task similar to random access codes and oblivious transfer. In the context of this task, the principle of information causality is stated as an inequality that defines the figure of merit of the task. This inequality is satisfied by physically allowed theories; on the other hand, if any nonsignaling correlation violates the CHSH inequality beyond Tsirelson's
bound, the information causality is violated. However, the information causality does not single out the whole set of quantum correlations among more general nonsignaling models. This follows because there are nonquantum correlations which do not violate the CHSH inequality beyond Tsirelson's bound, but satisfy the information causality~\cite{ABP+09}.

\section{Dimensionally restricted nonlocality}
Bell nonlocality is shown against any classical state of the two systems $\lambda$, whose dimension is not limited~\cite{DW15}. Thus, quantum Bell nonlocality implies nonlocality independently of the dimension of the state used inside the box. In this work, we consider dimensionally restricted nonlocality formalized in~\cite{Jeb25}, in which case nonlocality is shown against any classical state of the two systems $\lambda$, whose dimension is limited to the number of measurement results. Dimensionally restricted quantum nonlocality also occurs for certain Bell-local correlations. 

In~\cite{Jeb25},  a nonlinear determinant witness was shown to detect dimensionally restricted nonlocality. This witness was constructed in terms of  the covariance of  $A_x$ and $B_y$ given by  $\text{cov}(A_x,B_y)=\braket{A_xB_y} -\braket{A_x}\braket{B_y}$. $\braket{A_x}$ and $\braket{B_y}$ are marginal expectation values. The nonlinear witness denoted by ${NL}$ is given by
\begin{align}
NL=\left|\begin{array}{cc}\text{cov}(A_0,B_0) & \text{cov}(A_0,B_1)\\ 
\text{cov}(A_1,B_0) & \text{cov}(A_1,B_1) \end{array}\right|. \label{Wit}
\end{align}
A nonzero value of $NL$ witnesses dimensionally restricted nonlocality.
 In the framework of causal probabilistic theories, non-null discord is a generic feature of nonclassicality~\cite{Per12}. On the other hand, in the framework of generalized nonsignaling theories, dimensionally restricted nonlocality implies a generic feature of nonclassicality~\cite{Jeb25}. Therefore,  it becomes relevant to consider the question of whether any information-theoretic property associated with the generic feature can play a role in distinguishing quantum theory with information-theoretic limitations on quantum Bell nonlocality.  

 To address the aforementioned question, I consider the PR box fraction of dimensionally restricted nonlocality formalized in~\cite{Jeb25} as a measure of dimensionally restricted nonlocality. Using this measure of dimensionally restricted nonlocality,  how secure key distribution can also be achieved for certain Bell-local correlations has been studied.  To illustrate the idea of the PR box fraction of dimensionally restricted nonlocality, consider the noisy PR box,
\be \label{nPR}
P=p_{PR} P_{PR}+ (1-p_{PR})P_N,
\ee
where $0 \le p_{PR} \le 1$  and $P_N$ is the  maximally mixed box,  i.e., $P_N(ab|A_xB_y)=1/4$  for  all  $x,y,a,b$. The noisy PR box is Bell nonlocal for $p_{PR} >1/2$ since it  violates the CHSH inequality in this range; on the other hand, it has dimensionally restricted nonlocality for $p_{PR}>0$ since it takes the witness value $NL$ in Eq.~(\ref{Wit}) given by $NL=2p^2_{PR}>0$ for any $p_{PR}>0$.  The PR box fraction $p_{PR}$ in Eq.~(\ref{nPR}) is called the PR box fraction of dimensionally restricted nonlocality since its nonzero values measures dimensionally restricted nonlocality. Moreover, in~\cite{Jeb25}, this PR box fraction serves as a direct quantifier of the available resource for secure key generation. 

  To define the PR box fraction of dimensionally restricted nonlocality for correlations beyond the noisy PR box~(\ref{nPR}), 
 a nonlinear measure of correlations, $\Gamma$, defined in  Appendix~\ref{FPR}, was considered in~\cite{Jeb25}. Using this measure, it was shown in~\cite{Jeb25} that for any nonsignaling box, $P$, that can be decomposed as a convex mixture of a single PR box and a Bell-local box with $\Gamma=0$ as follows:
\begin{equation} \label{PRdecom} 
P=p_{PR} P^{\alpha\beta\gamma}_{PR} + (1-p_{PR}) P^{\Gamma=0}_{L},
\end{equation}
$4p_{PR}=\Gamma(P)$. Here, $P^{\Gamma=0}_{L}$ is a Bell-local box, with $\Gamma=0$.
For correlations which can be expressed  as in Eq.~(\ref{PRdecom}), $\Gamma(P)>0$ captures the PR box fraction of dimensionally restricted nonlocality beyond the noisy PR box~(\ref{nPR}).

\section{Results} 
 
\subsection{Certification of nonobjective information}

We consider the objective information of null discord states 
of causal probabilistic theories defined in~\cite{Per12} given by Definition \ref{defOIND}. Bell nonlocality certifies nonrealism~\cite{Bel64, GBK+07, Gis09}, which is the negation of the notion of realism by EPR~\cite{EPR35}. At the same time, certification of nonobjective information, which is the negation of the objective information in null-discord~\cite{Per12}, can be achieved as stated in the following proposition.
\begin{prop}
    Suppose that dimensionally restricted nonlocality ($NL>0$)
arises from a state $\rho$ in a causal probabilistic theory. Then it certifies
nonobjective information present in the state $\rho$.
\end{prop}
\begin{proof}
In the context of dimensionally restricted nonlocality, the outcome set $X$ in the definition of null discord states in Eq.~(\ref{nulldiscpt}) takes two values. It then follows that any Bell-local box arising from any of these null discord states has the form given by Eq.~(\ref{BLb}) with the dimension $d_\lambda$ of the state $\lambda$ bounded by $d_\lambda \le 2$~\cite{JAS17}. This implies that the dimensionally restricted nonlocality of any state in the causal probabilistic theory requires nonnull discord. Thus, if
dimensionally restricted nonlocality, i.e., $NL>0$, arises from a state $\rho$ in a causal probabilistic theory; it certifies nonobjective information of the state. 
\end{proof}

In the following, to distinguish quantum theory, I study generalized nonsignaling models, whose state space is a subpolytope of the full nonsignaling polytope, with the PR box fraction of nonobjective information together with the limitation of quantum Bell nonlocality. The PR box fraction of nonobjective information is defined as follows.
\begin{definition}
    A nonsignaling box given by Eq.~(\ref{PRdecom}) has the PR box fraction of nonobjective information if $NL>0$ and $\Gamma>0$. 
\end{definition}
 \subsection{Distinguishing quantum theory}
First, consider a nonsignaling model in which all nonlocal correlations are postquantum as considered in~\cite{RDB+19,BMR+19}. One such nonsignaling model has the state space given by
\be \label{GNSTpq}
P=c_0P_{PR}+ (1-c_0)  P_L,
\ee
with $0 \le c_0 \le 1$  and $P_L$ is a convex mixture of any four deterministic boxes that has the CHSH value $\mathcal{B}_{000}=2$.
Now I state the following lemma.
\begin{lem}\label{lempq}
In the nonsignaling model in which every state is given by Eq.~(\ref{GNSTpq}),
the PR box fraction of nonobjective information is equivalent to Bell nonlocality. On the other hand, all nonlocal correlations that lie below Tsirelson’s bound have postquantum models, and
the information causality is not sufficient to witness all these postquantum models.
\end{lem}
\begin{proof}
For any box $P$ given by Eq.~(\ref{GNSTpq}), the PR box fraction of nonobjective information is given by $\Gamma(P)=4c_0$ since $P_L$ in Eq. (\ref{GNSTpq}) has $\Gamma=0$ and $NL(P)>0$ for any $c_0>0$. On the other hand, any box $P$ that has the form~(\ref{GNSTpq}) is Bell nonlocal for any $c_0 >0$ since $\mathcal{B}_{000}(P) =2(1+c_0) > 2$ for any $c_0 >0$. This implies that the box is Bell nonlocal if and only if  $\Gamma>0$. This proves the first part of the lemma.

Note that any Bell nonlocal box in Eq.~(\ref{GNSTpq}) has postquantumness, which is known in~\cite{RDB+19}.
On the other hand, as studied in~\cite{BMR+19}, there are postquantum models in Eq.~(\ref{GNSTpq}) that lie below Tsirelson's bound, which are not indicated by the information causality.
\end{proof}

Consider a specific state space of~(\ref{GNSTpq}) given by
\be \label{GNSTpq1}
P=c_0P_{PR}+ (1-c_0)  (c_1 P^{0000}_D + c_2 P^{0101}_D),
\ee
with $0 \le c_0, c_1 \le 1$ and $c_1+c_2=1$.
For any Bell nonlocal box
of the form (\ref{GNSTpq1}), the information causality is violated ~\cite{RDB+19, BMR+19}. Before information causality was proposed to single out quantum theory, postquantumness of the noisy PR boxes that lie below Tsirelson bound was shown if the noisy PR boxes can be distilled into the PR box~\cite{BS09}. In~\cite{BMR+19}, it was shown that any noisy PR box of the form (\ref{GNSTpq1}) can always be distilled into the PR box. I note that the local state spaces of Eq.~(\ref{GNSTpq1}) have no simpliciality. 
Next, we consider another generalized nonsignaling  model whose local state spaces have no simpliciality. Still, it contains quantum Bell nonlocality, as indicated by Hardy's paradox~\cite{Har93}.
To define Hardy's paradox, consider the following conditions on the four joint probability distributions of the CHSH scenario: 
\begin{align}
\begin{split} \label{HC}
    P(01|A_0B_0)&=0, \\
     P(10|A_0B_1)&=0,\\
      P(10|A_1B_0)&=0,\\
       P(10|A_1B_1)&=p_{H}.
      \end{split}
\end{align}
If any given box satisfies the above equation with $p_H=0$, then the box is Bell-local. Otherwise, the box is Bell nonlocal with a success probability of Hardy's paradox $p_H>0$. The conditions of Hardy's paradox in Eq.~(\ref{HC}) have been defined such that $p_H>0$ implies that the CHSH inequality, $\mathcal{B}_{000}\le 2$, is violated by the box that has this Hardy's paradox; the other Hardy's paradoxes can also be defined corresponding to the violation of other CHSH inequalities~\cite{RDB+19}.
There exist quantum correlations that give rise to $p_{H}>0$~\cite{Har93}. In Ref. \cite{RZS12}, the analogue of Tsirelson's bound on $p_{H}$ was derived to be $5(\sqrt{5}-1)/2 \approx 0.09$, on the other hand, the PR box $P_{PR}$, which satisfies the Hardy's paradox in Eq.~(\ref{HC}), has the maximal success probability of $p_{H}=0.5$.  

 In the nonsignaling model of Hardy's paradox given by Eq.~(\ref{HC}), any nonsignaling box $P_{H}$ is given by
\begin{multline} \label{HP}
 P_{H}=h_{PR} P^{000}_{PR} + h_0 P^{0000}_D+ h_1 P_D^{0010} \\ 
+ h_2 P_D^{0101}+ h_3 P_D^{1101} 
 + h_4 P_D^{1110}.   
\end{multline}
Now I obtain the following lemma.
\begin{lem} \label{lemHP}
 In the nonsignaling model of Hardy's paradox,  the PR box fraction of nonobjective information is nonzero if and only if the box is Bell nonlocal. On the other hand, the information causality does not reproduce Tsirelson's bound of Hardy's paradox.
\end{lem}
\begin{proof}
For any nonsignaling box, $P_{H}$, given by Eq. (\ref{HP}), $p_{H}(P_{H})=\frac{h_{PR}}{2}$, on the other hand, $\Gamma(P_{H})=4h_{PR} >0$ if and if $p_{H}(P_{H})>0$. 

The bound on $p_{H}$ from the principle of information causality was derived in~\cite{AKR+10}. However, this bound does not reproduce Tsirelson's bound on $p_{H}$. This implies that there are nonquantum correlations that exhibit Hardy's paradox above Tsirelson's bound but are not ruled out by information causality as nonphysical correlations.
\end{proof}
I note that for the nonsignaling model given by Eq.~(\ref{HP}), the local state spaces are a simplex. Thus, the state space associated with Eq.~(\ref{HP}) is analogous to that of a causal probabilistic theory for which all nonnull discord states are entangled~\cite{Per12}.  Though the state space of Eq.~(\ref{HP}) has quantum Bell nonlocality, it does not fully capture the nonclassical state space of quantum theory. This follows because the local states of Eq.~(\ref{HP}) have no nonsimpliciality. This provides an intuition why the information causality does not indicate the limitation of quantum Bell nonlocality in the state space of Eq.~(\ref{HP})  to distinguish quantum theory.

Finally, we consider a  nonsignaling model whose local state spaces are nonsimplicial. Its state space is given by a polytope of a single PR box and all $16$ deterministic boxes. By introducing the concept of genuine boxes, this  nonsignaling model was considered in~\cite{SBP09} to study the emergence of quantum correlations in noisy PR boxes by nonlocality swapping. In the nonsignaling model of genuine boxes, a single CHSH inequality, which is maximally violated by the PR box present in the state space, is necessary and sufficient for implying Bell nonlocality.

I proceed to study the postquantumness of genuine boxes for a few specific families. We consider noisy PR boxes of the form
\be \label{ABP+}
P=\epsilon P_{PR}+ \nu Q + (1-\epsilon-\nu)P_N,
\ee
where $Q$ is one of the vertices, which are the other PR boxes except $P^{001}_{PR}$ and deterministic boxes, to be specified for the three families, and 
$P_N$ is the  maximally mixed box,  i.e., $P_N(ab|A_xB_y)=1/4$  for  all  $x,y,a,b$.
The postquantumness of the above noisy PR boxes was studied in Ref. \cite{ABP+09}. 

The first family that I study is the genuine boxes in the noisy PR boxes~ (\ref{ABP+}) given by
\be \label{genisoPR}
P=\epsilon P_{PR} +  \nu P^{100}_{PR}  + (1-\epsilon-\nu)  P_N,
\ee
which are genuine if $\nu \le 1/2$~\cite{SBP09}. 
The above family has the CHSH value given by $\mathcal{B}_{000}= 4\epsilon $, which implies that it is Bell nonlocal if and only if $\epsilon>1/2$. The Bell nonlocal boxes given in Eq.~(\ref{genisoPR}) have postquantumness if and only if $\epsilon^2 + \nu^2 >1/2 $~\cite{SBP09}, otherwise, they have a quantum model. I wish to note that the information causality is violated by these postquantum boxes even if the postquantumness is below Tsirelson's bound of the CHSH inequality~\cite{ABP+09}. 

The second family that I study is the genuine boxes in the noisy PR boxes~(\ref{ABP+}) given by
\be \label{genisoPR2}
P=\epsilon P_{PR} +  \nu P^{111}_{PR}  + (1-\epsilon-\nu)  P_N.
\ee
which are genuine if $\nu \le 1/2$ since Bell nonlocality of these noisy PR boxes is witnessed by the single CHSH inequality for $\nu \le 1/2$, otherwise, the other CHSH inequality witnesses Bell nonlocality.
The above family has the CHSH value given by $\mathcal{B}_{000}= 4\epsilon$, which implies that it is Bell nonlocal if and only if $\epsilon>1/2$ as in the case of the other family in Eq.~(\ref{genisoPR}). However, in contrast to the other family, the Bell nonlocal boxes of the family in Eq.~(\ref{genisoPR2}) have postquantumness indicated by the information causality if and only if $\epsilon >1/\sqrt{2} $~\cite{ABP+09}, otherwise, they have a quantum model~\cite{Jeb14}, which is, however, not indicated by the Navascue-Pironio-Acin criterion~\cite{NPA07} as illustrated in~\cite{ABP+09}. I wish to note that the postquantumness of the family~(\ref{genisoPR2}) does not lie below Tsirelson's bound of the CHSH inequality since the information causality is violated if and only if Bell nonlocality is above Tsirelson's bound~\cite{ABP+09}.

The third family that I study is the genuine boxes in the noisy PR boxes~(\ref{ABP+}) given by
\be \label{genisoPR3}
P=\epsilon P_{PR} +  \nu P^{0000}_{D}  + (1-\epsilon-\nu)  P_N.
\ee
which are genuine for any $\nu > 0$ since these noisy PR boxes are Bell nonlocal if and only if the single CHSH inequality is violated for any $\nu>0$.
The above family has the CHSH value given by $\mathcal{B}_{000}= 4\epsilon+2\nu$, which implies that the range in which it is Bell nonlocal is different than in that of the other two families in Eqs.~(\ref{genisoPR}) and ~(\ref{genisoPR2}). The Bell nonlocal boxes given in Eq.~(\ref{genisoPR3}) violate the information causality if and only if $(\epsilon + \nu)^2+ \epsilon ^2 >1 $~\cite{ABP+09}. On the other hand, there are Bell nonlocal boxes in Eq.~(\ref{genisoPR3})  that lie below Tsirelson's bound of the CHSH inequality and
do not violate the information causality~\cite{ABP+09}. 

Now I obtain the following lemma.
\begin{lem}\label{lem3}
In the nonsignaling model of genuine boxes~\cite{SBP09}, the state space admits the PR box fraction of nonobjective information. On the other hand,  the information causality identifies the physical limitation of quantum Bell nonlocality by Tsirelson's bound of the CHSH inequality, and postquantumness specific to the polytope of genuine boxes, which lie below Tsirelson's bound of the CHSH inequality, is also indicated by the information causality.
\end{lem}
\begin{proof}
To demonstrate that the nonsignaling model of genuine boxes has Bell-local boxes with the PR box fraction of nonobjective information, consider the genuine boxes in the noisy PR boxes~(\ref{ABP+}) given by
\be \label{isoPR}
P=\epsilon P_{PR} + (1-\epsilon)  P_N,
\ee
where $\epsilon$ satisfies $0\le \epsilon \le 1$.  For noisy PR boxes (\ref{isoPR}), the CHSH inequality is violated if and only if $\epsilon > 1/2$, on the other hand, the  PR box fraction of nonobjective information is given by $\Gamma=4\epsilon>0$ for any $\epsilon >0$. Thus, the state space of genuine boxes admits the PR box fraction of nonobjective information in Bell-local boxes.

Any genuine box that violates Tsirelson's bound of the CHSH inequality also violates the information causality~\cite{PPK+09}. On the other hand, Bell nonlocal boxes whose postquantumness lies below Tsirelson's bound and are specific to the state space of genuine boxes are also witnessed by the information causality. To show this, consider the noises acting on the noisy PR boxes in Eqs.~(\ref{genisoPR}),~(\ref{genisoPR2}) and~(\ref{genisoPR3}), 
\be  \label{PgB}
P=\nu Q + (1-\nu) P_N,
\ee
where $Q=P^{100}_{PR},P^{111}_{PR}$ and $P^{0000}_{D}$, respectively. I note
that in the case of the two families in Eqs.~(\ref{genisoPR}) and (\ref{genisoPR2}), the noises in Eq.~(\ref{PgB}) have the PR box fraction of nonobjective information indicated by $\Gamma=4\nu>0$ for any $\nu>0$, on the other hand, in the case of the family in Eq.~(\ref{genisoPR3}), the noise has $\Gamma=0$. This implies that 
for any $\nu >0$, the noisy PR boxes of the two families in Eqs.~ (\ref{genisoPR}) and~(\ref{genisoPR2}) have the specific feature of genuine boxes, on the other hand, the noisy PR boxes of the third family in Eq.~ (\ref{genisoPR3})  belong to the state space for which $\Gamma=0$. From this, I conclude that the postquantumness that lies below Tsirelson's bound and is specific to the genuine boxes is indicated by the information causality.   
\end{proof}

I now proceed to prove the following theorem.
\begin{thm}
In the context of the nonsignaling models for which all Bell-local boxes have $\Gamma=0$, such as given by Eqs.~(\ref{GNSTpq}) and~(\ref{HP}), the emergence of the PR box fraction of nonobjective information in Bell-local boxes of the polytope of genuine boxes~\cite{SBP09} isolates the postquantumness indicated by the information causality.
\end{thm}
\begin{proof}
The nonsignaling model of genuine boxes also contains all nonsignaling boxes in the other two nonsignaling models of Eqs.~(\ref{GNSTpq}) and~(\ref{HP}). In this context, I note that the postquantumness below Tsirelson's bound of the CHSH inequality, not indicated by the information causality, is due to the boxes or the noises that belong to a state space for which Bell nonlocality is equivalent to the PR box fraction of nonobjective information, as noted in Lemmas~\ref{lempq},~\ref{lemHP}, and~\ref{lem3}. 
Thus, in the context of the nonsignaling models for which all Bell-local boxes have $\Gamma=0$, the emergence of the PR box fraction of nonobjective information in Bell-local boxes in the nonsignaling model of genuine boxes isolates the postquantumness indicated by the information causality. 
\end{proof}

The main result is obtained in the following corollary of the above theorem.
\begin{cor}
Quantum theory is distinguished by the limitation of quantum Bell nonlocality in the nonsignaling model of genuine boxes, indicated by the information causality, and the emergence of certification of nonobjective information by the PR box fraction in the Bell-local boxes over the state spaces that do not have the PR box fraction of nonobjective information in Bell-local boxes.
\end{cor}
However, the above result only provides a partial answer to the question of what distinguishes quantum theory from generalized nonsignaling theories, as it is not without loss of full generality. This is because there are other nonsignaling models whose state space does not have the PR box fraction of nonobjective information in Bell-local boxes~\cite{ABP+09, AKR+10, Jeb14, BMR+19} that are not considered in the present work.   

\section{Conclusions} In summary, using an information-theoretic concept of certifying nonobjective information by the PR box fraction, I demonstrated that identifying the information-theoretic limitations of quantum Bell nonlocality alone is not sufficient to distinguish quantum theory from generalized nonsignaling theories. This demonstration follows from the studies in the present work that the limitations of quantum Bell nonlocality alone do not single out the full nonclassical state space of quantum theory. 
To obtain this conclusion,  I studied two nonsignaling models for which the PR box fraction of nonobjective information is equivalent to Bell nonlocality and a third nonsignaling model given by genuine boxes studied in Ref. \cite{SBP09}, which has the PR box fraction of nonobjective information in Bell-local boxes. I then demonstrated that in the case of genuine boxes, the emergence of the PR box fraction of nonobjective information in Bell-local boxes over the other nonsignaling models for which all Bell-local boxes do not have the PR fraction of nonobjective information has the following implication. It serves to isolate the postquantumness specific to the state space of genuine boxes by the information causality. This led to providing a partial answer to the question of what distinguishes quantum theory from generalized nonsignaling theories as follows. Quantum theory is distinguished by the limitation of quantum Bell nonlocality identified by the information causality in the state space of genuine boxes, together with the emergence of the PR box fraction of nonobjective information in Bell-local boxes over the other nonsignaling models that do not have the PR box fraction of nonobjective information in Bell-local boxes. However, astonishingly, this partial answer goes beyond the partial answer given by the information causality alone, as in~\cite{ABP+09}, or other partial answers that appeared before in~\cite{SBP09,BS09} by the nonlocality swapping and distillation protocols, respectively.  
I hope to answer the question completely in an upcoming complementary paper by using a three-way decomposition of the nonsignaling boxes in~\cite{Jeb14}, which I may present from the perspective of selftesting of quantum theory as explored in~\cite{WC20}. 

\section*{Acknowledgement} This work was supported by the National Science and Technology Council, the Ministry of Education (Higher Education Sprout Project NTU-113L104022-1), and the National Center for Theoretical Sciences of Taiwan.

\appendix 
\section{The nonlinear measure of correlations} \label{FPR}
Here, I define the nonlinear measure of correlations, denoted by $\Gamma$, consider in~\cite{Jeb25} to define the PR box fraction of dimensionally restricted nonlocality. This measure was constructed in terms of the  CHSH inequalities in the covariance form~\cite{PHC+17}.
Define  the  absolute covariance CHSH functions  $\texttt{cov}\mathcal{B}_{2\alpha+\beta}  :=
| \texttt{cov}(A_0B_0)        +         (-1)^{\beta}\texttt{cov}(A_0B_1)        +
(-1)^{\alpha} \texttt{cov}(A_1B_0)  + (-1)^{\alpha  \oplus  \beta \oplus  1}
\texttt{cov}(A_1B_1)|$. Consider the following triad of quantities constructed from these four covariance CHSH functions:
\begin{align}
\begin{split}
\Gamma_1&:= \Big||\texttt{cov}\mathcal{B}_0  - \texttt{cov}\mathcal{B}_1  | -
|\texttt{cov}\mathcal{B}_2  - \texttt{cov}\mathcal{B}_3|\Big|\\
\Gamma_2&:= \Big||\texttt{cov}\mathcal{B}_0  -\texttt{cov}\mathcal{B}_2  | -
|\texttt{cov}\mathcal{B}_1  - \texttt{cov}\mathcal{B}_3|\Big| \\
\Gamma_3&:= \Big||\texttt{cov}\mathcal{B}_0  -\texttt{cov}\mathcal{B}_3  | -
|\texttt{cov}\mathcal{B}_1  - \texttt{cov}\mathcal{B}_2|\Big|.\label{gi}
\end{split}
\end{align}
To capture the PR box fraction with dimensionally restricted nonlocality, the following quantity is defined:
\begin{equation}
\Gamma:= \min_i \Gamma_i. 
\label{eq:G}
\end{equation}
Here $\Gamma$ satisfies the following properties: (i) $ 0 \le \Gamma \le 4$; (ii) $\Gamma = 0$ for any product box of the form, $P(ab|A_xB_y)=P(a|A_x)P(b|B_y)$; (iii) $\Gamma$ is invariant under relabeling of inputs and/or outputs, and (iv)  $\Gamma= 4$ for any PR box $P^{\alpha\beta\gamma}_{PR}$.

\end{document}